\documentclass[12 pt]{article}
\usepackage{bbm}
\usepackage{natbib}
\usepackage{amssymb}
\usepackage{amsmath}
\usepackage{amsthm}
\usepackage[onehalfspacing]{setspace}
\usepackage[colorlinks]{hyperref}
\usepackage{graphicx}
\usepackage{caption}
\usepackage{subcaption}
\usepackage{enumitem}
\usepackage[margin=1in]{geometry}
\usepackage{verbatim}
\usepackage{sectsty}
\usepackage[svgnames]{xcolor}
\usepackage{subfiles}
\usepackage{xargs}
\usepackage[colorinlistoftodos,textsize=tiny]{todonotes}
\usepackage[capitalize,nameinlink]{cleveref}
\usepackage{mathtools}
\hypersetup{
    pdftitle={Reselling Information},
    pdfauthor={Nageeb Ali, Ayal Chen-Zion, Erik Lillethun},
    citecolor=blue,
    bookmarksnumbered=true,
    urlcolor=Indigo,linkcolor=DarkBlue
}
 \theoremstyle{remark}

\theoremstyle{plain}

\newtheorem{definition}{Definition}

\newtheorem{lemma}{Lemma}

\newtheorem{proposition}{Proposition}

\renewcommand{\epsilon}{\varepsilon}
\newcommand{\Agnts}{A}

\newcommand{\Infmd}{M}
\newcommand{\Uninfmd}{\Agnts \backslash \Infmd}

\usepackage{xcolor}
\usepackage{pgf}
\usepackage{tikz}
\usetikzlibrary{shapes,arrows,automata,fit}

\tikzstyle{info}=[circle,thick,draw=black,fill=black!25,minimum size=4mm]
\tikzstyle{uninfo}=[circle,thick,draw=black,fill=white,minimum size=4mm]
\tikzstyle{inforecog}=[circle,line width=1mm,draw=black!50,fill=black!25,minimum size=4mm]
\tikzstyle{uninforecog}=[circle,line width=1mm,draw=black!50,fill=white,minimum size=4mm]

\tikzstyle{traded}=[draw, line width=1mm]
\tikzstyle{recog}=[draw=black!50, line width=1mm]

\sectionfont{\color{DarkRed}}
\subsectionfont{\color{DarkRed}}
\subsubsectionfont{\color{DarkRed}}
\newcommandx{\nageeb}[2][1=]{\todo[linecolor=blue,backgroundcolor=blue!25,bordercolor=blue,#1]{#2}}
\newcommandx{\erik}[2][1=]{\todo[linecolor=red,backgroundcolor=red!25,bordercolor=red,#1]{#2}}

\begin{document}

\begin{titlepage}

\title{Reselling Information\thanks{We thank Mariagiovanna Baccara, Matt Elliott, Ben Golub, Navin Kartik, Bobby Kleinberg, Mihai Manea, Arnold Polanski, Jim Rauch, Joel Sobel, Eduard Talamas, Venky Venkateswaran, Joel Watson, and various audiences for useful comments, with special thanks to Guido Menzio for particularly valuable suggestions.}}
\author{
S. Nageeb Ali\thanks{Department of Economics, Pennsylvania State University. Email: \href{mailto:nageeb@psu.edu}{nageeb@psu.edu}.}
\and
Ayal Chen-Zion\thanks{Convoy Inc. This work was conducted prior to joining Convoy and represents the author's views alone. Email: \href{mailto:ayal.chenzion@gmail.com}{ayal.chenzion@gmail.com}.}
\and 
Erik Lillethun\thanks{Department of Economics, Colgate University. Email: \href{mailto:elillethun@colgate.edu}{elillethun@colgate.edu}.}
}
\maketitle

\begin{abstract}
\noindent 
Information can be simultaneously consumed, replicated, and sold to others. We study how resale affects a decentralized market for information. Even if the initial seller is an informational monopolist, she captures non-trivial rents from at most a single buyer in any Markovian equilibrium: her payoffs converge to $0$ as soon as a single buyer has bought information. By contrast, there exists a non-Markovian ``prepay equilibrium'' where payment is extracted from all buyers before the information good is released. By exploiting resale possibilities, this prepay equilibrium gives the seller as high a payoff as she would achieve if resale were prohibited. \vspace{0.4in}

\end{abstract}
\thispagestyle{empty} 

\end{titlepage}
\clearpage
\setcounter{page}{1}
%
%

\section{Introduction}
As emphasized by \citet{arrow1962economic} (and others before him), information is a difficult commodity to trade for at least two reasons. First, one may be unable to prove that one has valuable information without revealing it. Second, information may be replicated and re-sold once it has been purchased. Our interest is in that second problem: 
When a good is replicable, can a seller charge a high price for it?\footnote{To resolve the first problem, \cite{anton1994expropriation} offer a contracting solution where the information monopolist commits to sell information to competing buyers if a buyer steals her idea without payment; \cite{horner2016selling} offer a non-contractual solution that involves gradually selling information and collecting payments.}

We pose this question in a model of decentralized bargaining where a monopolist for information may have significant bargaining power, can personalize prices, negotiates bilaterally with buyers who face frictions in negotiations, and can coordinate behavior on her favorite equilibrium of this game. We find that nevertheless, she appropriates very little of the social surplus from selling information across all Markovian equilibria; a robust upper-bound for her payoff is the value of information of a \emph{single} buyer. This negative result suggests that without protection from resale, sellers of information have little incentive to acquire information even if they have significant market power.

Why does the seller fail to appropriate much of the social surplus? The challenge is that of commitment: neither a seller nor buyers can commit to not sell information to third-parties in the future. Thus these agents anticipate future competition. This commitment problem endemic to a market for information is both dynamic and multilateral. Buyers are unwilling to buy information at a high price today if they can buy it cheaply in the future. The reason that they can do so is that those who have information cannot commit not to compete with each other. If agents were able to commit---either if the initial seller had exclusive rights to sell information (e.g., as in a copyright with a permanent duration), or if each party could commit to selling information only once---then the seller may appropriate a substantial fraction of social surplus. The information good monopolist is stymied not only by the buyers' ability to resell information but also her own ability to do so.

We find a non-contractual solution to this commitment problem that exploits resale. The seller uses a ``prepay equilibrium,'' where she holds up the release of the information good until she has collected prepayments (unconditional one-sided monetary transfers) from all but one buyer. Using these prepayments and delay, the seller achieves the same payoff that she would if resale were prohibited. In this equilibrium, buyers pay substantial amounts in prepayments because prepayment allows them to buy information in the future at very cheap prices; those cheap prices obtain precisely because of the resale problem alluded to above. Effectively, the seller exploits her and others' commitment problem to make it credible for buyers to pay before they receive information. Thus, within the commitment problem lies the seed to its solution.

\paragraph{Framework and Results:} 
We study decentralized Nash bargaining in a market where a single seller is connected to multiple buyers. At time $0$, that single seller has information, which is valuable to each buyer. The information may correspond to knowing a payoff-relevant state of the world or a transferable skill; alternatively, it could be some replicable content that the seller has produced (and cannot protect with a copyright). The seller and each buyer face a trading opportunity at a random time. At a trading opportunity, the pair bilaterally transfer money and/or information (or do nothing) in a way that maximizes their joint continuation value. If information trade occurs, the buyer consumes the information and pays a price. But the challenge is that the buyer can resell this information to others. Thus, in selling information to a buyer, the seller creates her own competition.  

We study behavior at the frequent-offer limit, where the time between trading opportunities converges to $0$. \Cref{Proposition-TwoInformedPlayers} shows that prices converge to $0$, robustly across all Markovian equilibria, as soon as two parties possess the information, i.e., once a single buyer has purchased information. The intuition for this result is subtle because there are two opposing forces. On the one hand, no seller wishes to lose the opportunity to sell to a buyer, and so competition between sellers reduces the price of information. But on the other hand, our decentralized framework features trading frictions in which each buyer meets at most one seller at any instance, and bears some delay in waiting for the next trading opportunity. Analogous to the Diamond Paradox \citep{diamond1971model}, one may anticipate that these trading frictions would benefit sellers. In our setting, the competitive effect dwarfs the ``Diamond effect'' so that prices converge to $0$. 

Given this future competition, the monopolist captures rents in a Markovian equilibrium only from her first trade. We show that the seller-optimal among these equilibria uses delay as a strategy to capture some rents: the seller gets a high price for her first sale of information and there is a particular buyer designated to be her exclusive ``first buyer.'' After that first sale, prices converge to $0$. Because a designated buyer knows that he, in equilibrium, cannot obtain information from any other source, it is as if he and the seller are negotiating bilaterally for that information, and thus, the price is bounded away from $0$. Given this possibility, we show that it is self-enforcing for the seller and every other buyer to disagree until that first sale is made. While this strategy offers the seller some surplus, the seller can do this only once. Thus, across Markovian equilibria, resale severely limits the seller's ability to capture rents. 

We use this equilibrium to construct our prepayment scheme. 
When there are $n_B$ buyers, the seller trades the information good itself and also collects prepayments. She first collects prepayments from $n_B-1$ buyers and then trades the information good exclusively to the buyer who did not prepay before the others. We show that the seller can then appropriate approximately the same surplus from each buyer as if she were bargaining bilaterally with each and information could not be resold. 

Here is the strategic logic. Each buyer knows that he must either prepay or be the exclusive first information good buyer. In the latter case, he gives up substantial surplus to the seller, as in the seller-preferred Markovian equilibrium. Prepaying avoids this latter case, and lets him obtain information in the future at a price of approximately $0$. Thus, in equilibrium, prepayment becomes almost as valuable as information itself. Because prepayment opportunities are scarce, the seller extracts the entire seller's share of the value of prepayment from each buyer. The seller herself faces no incentive to deviate and sell information early because it impedes her ability to extract surplus from future buyers. 

This prepayment scheme requires no commitment: it is simply a non-Markovian equilibrium where keeping track of past prepayments adds just the right amount of history-dependence to solve the commitment problem. The resolution \emph{isn't} through a punishment scheme (in the spirit of repeated games) but instead exploits (a) bargaining dynamics when there are fewer prepayment opportunities than buyers, and (b) information prices converging to $0$ once a single buyer has bought information. It offers a non-contractual alternative to intellectual property protection. 

\paragraph{Related Literature:}

We contribute to the vast literature on information markets and intellectual property. One view, dating back to at least \cite{schumpeter} and also in \cite{grossman1994endogenous}, is that imperfect competition and bargaining are necessary ingredients for people to have incentives to acquire information. Our negative result shows that even if a seller is a monopolist and has substantial bargaining power, she may be unable to capture much of the social surplus that comes from acquiring information without restrictions to resale. But our positive result shows that an information-good monopolist can capture much of these rents through the use of both prepayments and delay. 


The resale-commitment problem is the focus of the innovative (and, in our view, underappreciated) study of \citet{polanski2007decentralized}. He studies information resale on arbitrary networks in an environment without discounting and restricts attention to an immediate agreement equilibrium. In this equilibrium, prices equal $0$ along any cycle in the graph. Also studying this setting, \cite{manea2021bottleneck} provides a complete payoff characterization (for the undiscounted limit) of the immediate agreement equilibrium in terms of the global network structure for general networks. In his setting, buyers may be heterogeneous in their intrinsic value for information, and some buyers may not value information intrinsically at all. His analysis uses ``bottlenecks'' and ``redundant links'' to offer an elegant graph-theoretic perspective of competition and intermediation in the networked market.

Our work complements these prior papers. Instead of restricting attention to the immediate agreement equilibrium in the undiscounted game, we derive bounds on prices across \emph{all} Markov equilibria in the frequent-offers limit of games with discounting.\footnote{One reason to study a frequent-offers limit of the game rather than the undiscounted game directly is that it is not obvious that the equilibrium correspondence is continuous at the undiscounted limit.} Looking at other equilibria is important because the seller-optimal equilibrium exploits delay. Our negative result shows that sellers' inability to capture significant rents emerges in all Markovian equilibria. Moreover, our solution to the resale problem, where we combine delays with prepayments in a non-Markovian equilibrium, has no counterpart in prior work.

Other settings offer related albeit distinct intuitions for how innovators may appropriate some value of their information. \cite{baccara2007bargaining} focus on the challenge of information leakage, where an innovator who hires agents to bring a product to market has to reveal the idea before the agent agrees to the offer. They show that an innovator can nevertheless appropriate a large share of the surplus with the threat of competition in the product market if the agent rejects the offer. Unlike our setting, future competition is an off-path threat that deters leakage; by contrast, in our analysis, it is the anticipation of (on-path) future competition and resale that incentivizes buyers to prepay for information. In \cite{henry2011waiting}, a monopolist induces potential imitators to delay entry through the credible promise that once information is obtained by one imitator, the resale market will ensure that information is cheap subsequently. Here, the equilibrium induces a war of attrition that benefits the initial monopolist. Our prepay equilibrium features different dynamics where each buyer attempts to preempt other buyers by prepaying for information.


%

Our stylized analysis omits a number of important features. One of these is direct ``consumption-externalities'' for buyers, which has been studied extensively in the context of financial markets \citep{admati1986monopolistic}. \citet{muto1986information} studies how these negative externalities from others holding information can deter buyers from reselling information. \cite{polanski2019communication} extends results from his prior work to illustrate how they apply with consumption externalities. 

\section{Examples}\label{Section-Example}
We illustrate the main ideas of our paper using a simple example. Seller $S$ (``she'') has information that is valuable to two buyers, $B_1$ and $B_2$ (each of whom is a generic ``he''), and is worth $1$ util to each buyer. The agents share a discount rate of $r$, and each link meets with probability $\approx\lambda \,dt$ in a period of length $dt$. 
When a pair meets, transfers are determined through symmetric Nash Bargaining where agents' outside options are their continuation values without trade occurring, but following the same equilibrium. The ratio  $\lambda /r$ measures the frequency of trading opportunities per unit of effective time, and we take $\lambda/r\rightarrow\infty$ as the frequent-offers limit of the game.
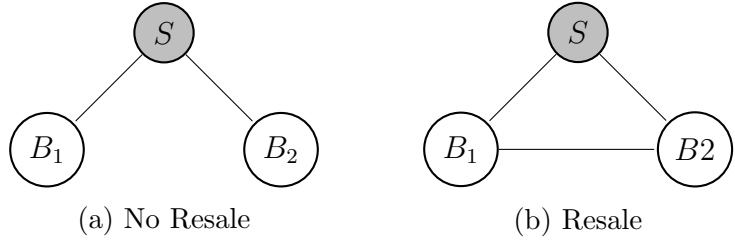
\begin{figure}[h!]
\centering
\begin{subfigure}{0.3\textwidth}
\centering
\begin{tikzpicture}[shorten >=1pt,node distance=2.2cm,auto]
\node[info] (1) {$S$};
\node[uninfo] (2) [below left of=1] {$B_1$};
\node[uninfo] (3) [below right of=1] {$B_2$};
\path[-] (1) edge [above left] node {} (2);
\path[-] (1) edge [above right] node {} (3);
\end{tikzpicture}
\caption{No Resale\label{Figure-Isolation}}
\end{subfigure}
\quad
\begin{subfigure}{0.3\textwidth}
\centering
\begin{tikzpicture}[shorten >=1pt,node distance=2.2cm,auto]
\node[info] (1) {$S$};
\node[uninfo] (2) [below left of=1] {$B_1$};
\node[uninfo] (3) [below right of=1] {$B_2$};
\path[-] (1) edge [above left] node {} (2);
\path[-] (1) edge [above right] node {} (3);
\path[-] (2) edge [below] node {} (3);
\end{tikzpicture}
\caption{Resale\label{Figure-Competition}}
\end{subfigure}
\caption{A Single Seller Trades Information with $2$ Buyers.}
\end{figure}\label{Figure-Example}

\subsection{No Resale: A Benchmark}\label{Section-NoResale}
First, consider the setting depicted in \autoref{Figure-Isolation}: each buyer can buy information from the seller but cannot resell information to the other buyer. The benchmark therefore corresponds to two separate bilateral bargaining games. In each game, the symmetric Nash Bargaining price evenly splits the gains from trade between the buyer and seller. These gains are dynamic: the gain from trade is the gap between the joint surplus from trading today and from waiting for the next opportunity. Hence, the price solves 
\begin{align}\label{Equation-NoResale}
\underbrace{p - \frac{\lambda}{r + \lambda} p}_{\text{Seller's Gain from Trading Today}} =\underbrace{(1-p) - \frac{\lambda}{r + \lambda} (1 - p)}_{\text{Buyer's Gain from Trading Today}},
\end{align}
where $\frac{\lambda}{r + \lambda} = \int_0^\infty e^{-rt}e^{-\lambda t}\lambda\, \,dt$ is the factor representing the cost of delay (the pair's next trading time is distributed exponentially). 
The solution, $p=\frac{1}{2}$, assures that the seller obtains half of the social surplus.  


\subsection{An Immediate Agreement Equilibrium with Resale}
\label{Section-AgreementEqmExample}
We now consider the setting shown in \autoref{Figure-Competition} where all three agents are connected and can trade information. We first consider an immediate agreement equilibrium where parties expect every trading opportunity to result in an information trade. For this discussion, it is useful to define $\gamma\equiv \int_0^\infty e^{-rt}e^{-2\lambda t}\lambda\, dt = \frac{\lambda}{r + 2 \lambda}$.
For an agent, fixing one of his links $\ell$, the term $\gamma$ is the discounted weight that of his two links, $\ell$ is the next link that is selected. In the frictionless limit, as $\frac{\lambda}{r}\rightarrow\infty$, $\gamma$ approaches $\frac{1}{2}$. 
We first consider the history after only one buyer has bought information, and then consider the original game in which no buyer has bought information. \smallskip

\noindent\emph{Only one buyer is informed}: Suppose that buyer $B_1$ is informed but buyer $B_2$ is not; the latter can then purchase information from one of two parties. Denote the equilibrium price in this history by $p(2)$ (because there are two agents who can sell information). Symmetric Nash Bargaining implies that $p(2)$ solves
\begin{align}\label{Equation-OneBuyerInformed}
\underbrace{p(1-\gamma)}_{\text{Seller's Gain from Trading Today}}=\underbrace{(1-p)(1-2\gamma).}_{\text{Buyer's Gain from Trading Today}}
\end{align}
The current seller's gain from trading today is that she secures the sale; by contrast, if she waits, she can only sell to the uninformed buyer if that buyer does not meet the other seller. Therefore, for any strictly positive price $p$, the LHS is at least $p/2$. By contrast, the buyer's only gain from waiting is avoiding delay, which vanishes at the frictionless limit. For \cref{Equation-OneBuyerInformed}  to hold as $\frac{\lambda}{r}\rightarrow\infty$, the price $p(2)$ converges to $0$.\footnote{The equilibrium price, $p(2)$, equals $p(2)=\frac{1-2\gamma}{2-3\gamma}$, which converges to $0$ as $\frac{\lambda}{r}\rightarrow\infty$. This outcome replicates Bertrand competition even though the buyer never meets both sellers simultaneously.}  
\smallskip

\noindent\emph{No buyer is informed}: Low second stage prices influence negotiations at the earlier stage before any buyer has purchased information. The payoff of buying information at price $p$ for the first buyer then is an immediate benefit of $(1-p)$, and the potential for reselling that information if he meets the other buyer first, which has a discounted value of $\gamma p(2)$. His ``outside option'' is his payoff from waiting and not trading today: in the future, either the seller meets him first in which case he is in the same position as today, or the seller meets the other buyer first, in which case he can buy information at a price of $p(2)$ from either of the other two agents. The first contingency obtains a stochastic discount of $\gamma$ (reflecting that the seller does not meet the other buyer first), and the second contingency obtains a stochastic discount of $2\gamma^2$ (reflecting that the other buyer buys information first, after which there are twice as many opportunities to buy information). Thus, his gain from buying information at any fixed price $p$ is 
\begin{align*}
	1-p+\gamma p(2)-\gamma \left(1-p+\gamma p(2)\right)-2\gamma^2 (1-p(2))\rightarrow -\frac{p}{2},
\end{align*}
which implies that for trade to occur, the price that the first buyer pays for information must converge to $0$. The intuition here is that each buyer recognizes that he can wait to be the second buyer and obtain the information for virtually free; hence, he has no gain from securing the information now at a strictly positive price.

\subsection{The Seller's Optimal Equilibrium (Markovian)}\label{Section-OptimalEqmExample}
The issue above is that all parties expect every trading opportunity to end with information being traded. We now show how the seller can credibly use delay and disagreement to obtain $1/2$ from the first buyer. In this equilibrium, only one specific buyer anticipates having the opportunity to be the second purchaser of information. Suppose that the seller never trades information with buyer $B_2$ \emph{until} she has sold information to buyer $B_1$; after she sells information to $B_1$, then she and $B_1$ shall compete to sell information to $B_2$ at the price $p(2)$ characterized above. Thus, buyer $B_1$ never anticipates being able to be the second purchaser of information. The price that the seller and $B_1$ agree to solves
\begin{align}\label{Equation-EndogenousBottleNeckExample}
\left(p^*+\gamma p(2)\right)\left(1- \frac{\lambda}{r + \lambda} \right)=\left(1-p^*+\gamma p(2)\right )\left(1- \frac{\lambda}{r + \lambda} \right).
\end{align}
The LHS reflects the seller's gain from trading today versus waiting for tomorrow: she obtains price $p^*$ today and after selling information, she can potentially be the first to sell information to the other buyer (which is thus discounted by $\gamma$). The RHS is buyer $B_1$'s gain from trading today versus waiting. If he agrees, then he obtains a payoff of $1-p^*$ today and potentially resells information successfully at $p(2)$ (discounted by $\gamma$). If he rejects, he anticipates that the next trading opportunity where there may be trade is where once again, he faces the prospect of being the first to buy information from the seller. Solving \cref{Equation-EndogenousBottleNeckExample} yields that $p^*= \frac{1}{2}$.

For this to be an equilibrium, the seller and buyer $B_2$ must not trade until information has been sold to $B_1$. This is compatible with Nash Bargaining only if 
\begin{align*}
\underbrace{\frac{\lambda}{r+\lambda}\left(p^*+\gamma p(2)\right)}_{\text{Seller's Surplus with No Trade}}+\underbrace{\frac{\lambda}{r+\lambda}\left(2\gamma (1-p(2))\right)}_{\text{Buyer's Surplus with No Trade}}\geq \underbrace{1+2\gamma p(2).}_{\text{Joint Surplus with Trade}}
\end{align*}
As $\frac{\lambda}{r}\rightarrow \infty$, the LHS converges to $\frac{3}{2}$ whereas the RHS converges to $1$. Therefore, the seller and buyer $B_2$ would not trade before the seller sells information to $B_1$.

\subsection{A Prepay Equilibrium (non-Markovian)}
We now study how the seller can do better by collecting unconditional one-sided prepayments. The seller allows only one buyer to prepay.  Once a buyer has prepaid, then we exploit the equilibrium constructed in \Cref{Section-OptimalEqmExample} where the other buyer has to be the first to buy information. Thus, if buyer $B_i$ prepays, the seller sells information first only to the other buyer $B_j$. Once $B_j$ purchases information, then the seller and $B_j$ compete to sell information to $B_i$, the buyer who had prepaid.\footnote{In an off-equilibrium path history where the seller sells information before any prepayment, the seller and buyer who has bought information compete in selling it to the other buyer.} 
Prepayment effectively gives a buyer the ability to be the second purchaser of information.

In this equilibrium, once one buyer has prepaid, the other buyer pays a price of $p^*$ for information (as in \Cref{Section-OptimalEqmExample}). Now, consider the initial prepayment stage. Disagreement means that the seller may next meet with either buyer to try to collect prepayment. The price $p^{**}$ for prepayment equalizes gains from trade to the seller and buyer, and thus solves
\begin{multline}\label{Equation-PrepayExample}
\left(1- \frac{\lambda}{r + \lambda} \right) \left( p^{**}+ \frac{\lambda}{r + \lambda} p^{*} + \gamma p(2) \right) 
= \begin{aligned}\left(1- \gamma \right)& \left(-p^{**} + \frac{\lambda}{r + \lambda} 2 \gamma [1 - p(2)] \right ) \\- \gamma & \left( \frac{\lambda}{r + \lambda} [1 - p^{*} + \gamma p(2)] \right)\end{aligned}.
\end{multline}
The LHS is the seller's gain from collecting prepayment now versus waiting for the next trading opportunity, and the RHS is the buyer's gain.\footnote{Let us explain how each side of this equation is derived. On the LHS,  the seller first collects prepayment at price $p^{**}$, then after some delay sells information to the other buyer at price $p^{*}$, and finally may get the opportunity to sell information to the remaining uninformed buyer at price $p(2)$; her disagreement payoff is the same as her agreement payoff, subject to some delay. On the RHS, if the buyer agrees, he pays $p^{**}$, then waits until information is sold to the other buyer, and finally pays $p(2)$ for information from either of the other agents. If he disagrees, there are two possibilities. The first possibility is that the next trading event is the same and leads to the same payoff as the agreement payoff; this explains the first term, in which the $(1-\gamma)$ discount reflects the gain from agreement today versus waiting for this event in the future. The other possibility is that the seller collects prepayment from the other buyer; then, this buyer anticipates paying $p^*$ for information in the future and possibly selling it to the other buyer (the one who prepaid) at price $p(2)$. } Because $p^*=1/2$ and $p(2)\rightarrow 0$, we see that the price of the prepayment, $p^{**}$ converges to $\frac{1}{2}$ in the frequent-offer limit. Thus, the seller receives approximately $1/2$ for one buyer's prepayment and $1/2$ for selling information to the other buyer. 

The logic here is that prepaying gives a buyer the opportunity to buy information as the second purchaser, at which point, the price for information is $\approx 0$. By contrast, if the other buyer prepays, one has to purchase information at a price of $\approx 1/2$. Thus, the equilibrium value of prepayment is $\approx 1/2$. As prepayment is available to only one buyer, competition between the two buyers implies that the seller can extract close to $1/2$ for prepayment in the frequent-offers limit. Rather than being stymied by future resale possibilities, this prepay equilibrium exploits them to secure rents for the seller. Observe that, in the frequent-offer limit, the seller's payoff here coincides with that of \Cref{Section-NoResale} (where resale was prohibited).
\section{Model}\label{Section-Model}

\paragraph{Environment:}
A set of \emph{initial} buyers $ B\equiv \{b_1,\ldots,b_{n_B}\}$ interacts with a set of \emph{initial} sellers $ S\equiv\{s_1,\ldots,s_{n_S}\}$. The set of all agents is $A\equiv  B\cup S$ and $n\equiv n_B+n_S\geq 3$ is its cardinality. Each seller has identical information of value $v>0$ to each buyer. Information is replicable so it can be sold separately to each buyer, and any buyer who gains it can resell it to others. Hence, the set of buyers and sellers changes over time, and so we distinguish between ``informed'' and ``uninformed'' agents. 

Each pair of agents $i$ and $j$ has a link---a set $\{ i, j \}$, usually denoted $ij$ or $ji$---which reflects the potential to transfer money and/or information. Information may be transferred only from an informed agent to an uninformed one, and money may be transferred between an uninformed agent and an informed one.\footnote{In principle, one could allow money to change hands between pairs of informed agents or pairs of uninformed agents, but this adds notation without changing our results.} We denote the complete network by $G \equiv \{ \{ i, j \} \in 2^{A} | i \neq j \}$.

Let $T\equiv[0,\infty)$ denote the interval of \emph{real-time}. Trading opportunities occur at discrete points of time in  $T_\Delta\equiv\{0,\Delta,2\Delta,\ldots\}$, where $\Delta>0$ is the period length between trading opportunities. Time periods are indices $t \in \{ 1, 2, \ldots \}$, each beginning at real-time $(t - 1) \Delta$. At each time period, a pair of agents $i$ and $j$ is selected uniformly at random to potentially trade money and/or information, independently of the past; the probability that a particular pair is selected is therefore $\rho \equiv \frac{2}{n (n - 1)}$. Each time that a link is selected, it can result in one of three outcomes: $D$ denoting disagreement, $A_M$ denoting an agreement to transfer money but not information, and $A_I$ denote both information and potentially money. Therefore, the set of all histories is $H \equiv \{ \{ ( i_{\tau} j_{\tau}, \alpha_{\tau}, p_{\tau} ) \}_{\tau = 1}^{t} : t \in \mathbb{N}^{+}, i_{\tau} j_{\tau} \in G, \alpha_{\tau} \in \{ D, A_{M}, A_{I} \}, p_{\tau} \in \Re, \forall \tau \}$, where $ij$ is the link selected in that period, $\alpha \in \{ D, A_{M}, A_{I} \}$ denotes the outcome, and $p$ is the money transferred from the uninformed agent to the informed agent, or $0$ if both agents are the same type. For any history $h$ and any single period history entry $(ij, \alpha, p)$, $\{ h, (ij, \alpha, p) \}$ is the history formed by appending $(ij, \alpha, p)$ to the end of sequence $h$. Associated to each history is the set of agents who have (and can therefore sell) information at that history; for history $h$, we denote this set by $M(h)$. We define $\mathcal \Infmd\equiv\left\{\Infmd\subseteq A: S\subseteq M \right\}$ as the set of feasible sets of informed agents. Let $m = | M |$ be the number of informed agents when the set $M$ is clear from the context.

Agents have discount rate $r\geq 0$ with a per-period discount factor of $\delta \equiv e^{-r\Delta}$. Our results concern the \emph{frictionless limit} of the market, namely $\Delta\rightarrow 0$. 

\paragraph{Solution Concept:}
Our solution concept combines \emph{rational expectations} in how agents derive continuation values from future actions and \emph{Nash Bargaining}, which specifies their actions given those continuation values. We define this solution concept broadly to encompass both ``Markovian equilibria,'' where all decisions condition on only the set of informed agents, as well as the prepayment equilibria.%

%
We define a value function $V:A\times H\rightarrow \Re$ such that $V(i,h)$ is agent $i$'s expected payoff after the history $h$. It is measured from the start of a period, before a pair is selected for a trading opportunity. \emph{Trading functions} comprise prices $p: G \times H \rightarrow \Re$ and decisions $\alpha: G \times H \rightarrow \{D,A_{M},A_{I}\}$ such that $\alpha(i,j,h)=D$ if there is disagreement, $\alpha(i,j,h)=A_{M}$ if there is agreement but no information is transferred, and $\alpha(i,j,h)=A_{I}$ if information is transferred, when the history is $h$. If no information is exchanged, money exchanges are still captured by $p(i, j, h)$, the money transfer from the uninformed agent to the informed agent ($0$ if agents are the same type).
\cref{Definition-RationalExpectations} uses trading functions to derive equilibrium continuation values while \cref{Definition-NashBargaining} uses these continuation values to derive trading functions consistent with Nash Bargaining. \Cref{Definition-Equilibrium} brings these two together to define an equilibrium. 

\begin{definition}\label{Definition-RationalExpectations}
Given trading functions $(\alpha,p)$, $V$ satisfies \textbf{Rational Expectations} if for every history $h$, for every informed agent $s$ and uninformed agent $b$, 
\footnotesize
\begin{align*}
\begin{split}
V(s, h)\equiv & \rho \Bigl[ \underbrace{\sum_{b' \in A \setminus M(h)} p(s,b',h)}_{\text{Instantaneous Gains from Trade}}+\underbrace{\sum_{i j \in G} \delta V(s, \{ h, (ij, \alpha(i, j, h), p(i,j,h)) \})}_{\text{New Continuation Value}} \Bigr],
\end{split} \\
\begin{split}
V(b,h)\equiv & \rho \Bigl[ \underbrace{\sum_{s' \in M(h)} [ \mathbbm{1}\{ \alpha(s', b, h) = A_{I} \} v - p(s',b,h)]}_{\text{Instantaneous Gains from Trade}}+\underbrace{\sum_{i j \in G} \delta V(b, \{ h, (ij, \alpha(i, j, h), p(i,j,h)) \})}_{\text{New Continuation Value}} \Bigr],
\end{split}
\end{align*}
\normalsize
\end{definition}
The above expressions, taking behavior as given, describe an agent's expected payoff from potential transfers and associated changes in continuation values, as well as changes in continuation value based on other interactions that occur in that time period. Our next definition uses these continuation values to describes behavior consistent with Nash Bargaining. 

There are two key features of Nash Bargaining: first, trading information or money happens within a pair only if it increases their joint (bilateral) surplus, and second, the transfers split those gains from trade in a particular ratio. Symmetric Nash bargaining splits those gains evenly. We allow for asymmetries in bargaining power that depend on whether a party is buying or selling information. To that end, we denote by $w\in (0,1)$ an informed agent's---i.e., seller's---bargaining weight.

We first derive Nash Bargaining prices given trading decisions. Suppose that at history $h$, an informed agent $s$ and uninformed agent $b$ meet. The trading decisions at this stage could involve disagreement ($\tilde\alpha=D$), an agreement to transfer money but not information ($\tilde\alpha=A_M$), or an agreement to transfer both information and potentially money ($\tilde\alpha=A_I$). For every agreement decision, $\tilde{\alpha}\in \{A_M,A_I\}$, the price function $P(\tilde\alpha)$ is equal to the price $\tilde p$ that solves
\begin{align}\label{Equation-pricefunction}
\frac{w}{1-w}=\frac{ \overbracket{ \tilde{p} + \delta  V(s, \{ h, (bs, \tilde{\alpha}, \tilde{p}) \})- \delta  V(s, \{ h, (bs, D, 0) \}) }^{\text{\emph{Change in Seller's Surplus}}}}{ \underbracket{ \mathbbm{1}\{ \tilde{\alpha} = A_{I} \}v - \tilde{p}+ \delta  V(b, \{ h, (bs, \tilde{\alpha}, \tilde{p}) \}) - \delta  V(b, \{ h, (bs, D, 0) \})}_{\text{\emph{Change in Buyer's Surplus}}}}.
\end{align}
This equation states that the gains from trade are split according to the ratio $w:(1-w)$. There is an immediate gain of $v$ to the buyer if the parties agree to trade information ($\tilde\alpha=A_I$). Moreover, the seller's and buyer's continuation values may change relative to those from disagreement if either money or information is transferred. Allowing for such changes in continuation value allows us to model the history-dependence exploited by prepayments. 

\Cref{Equation-pricefunction} describes prices for agreement decisions; although prices will not matter for disagreement decisions ($\tilde\alpha=D$), for concreteness, we set $P(D)=0$.

The above describes prices for arbitrary trading decisions; Nash Bargaining also stipulates that trading decisions should maximize (bilateral) joint surplus at an interaction. This implies that at a history $h$, the trading decision $\alpha(s,b,h)$ should select the value of $\tilde\alpha\in \{D,A_M,A_I\}$ that maximizes
\begin{align}\label{Equation-AgreementDecisions}
\mathbbm{1}\{ \tilde{\alpha} = A_{I} \} v + \delta V(s, \{ h, (bs, \tilde{\alpha}, P(\tilde{\alpha})) \}) + \delta V(b, \{ h, (bs, \tilde{\alpha}, P(\tilde{\alpha})) \} ).
\end{align}
The above denotes the bilateral joint surplus following any of the trading decisions, given that the prices that are set correspond to Nash Bargaining prices. 

Our definition below imposes the above as Nash Bargaining trading functions.

\begin{definition}\label{Definition-NashBargaining}
Given a value function $V$, trading functions $(\alpha,p)$ satisfy \textbf{Nash Bargaining} if for all histories $h$, for all informed agents $s$ and uninformed agents $b$,  the trading decisions $\alpha(s,b,h)$ maximize \eqref{Equation-AgreementDecisions} and the prices correspond to \Cref{Equation-pricefunction}, viz., $p(s,b,h)=P(\alpha(s,b,h))$, and for all $i, j$ who are both uninformed or both informed, $\alpha(i, j, h) = D$ and $p(i, j, h) = 0$.
\end{definition}

Connecting continuation values and trading functions defines our solution concept. 

\begin{definition}\label{Definition-Equilibrium}
An \textbf{equilibrium} is a triple $(\alpha,p,V)$ such that given $(\alpha,p)$, the continuation value function $V$ satisfies Rational Expectations and given $V$, the trading functions $(\alpha,p)$ satisfy Nash Bargaining.
\end{definition}
This solution concept combines Nash Bargaining with dynamic considerations.  \Cref{Section-Trading} characterizes equilibria where all decisions condition only on the set of informed agents $M$, which is the payoff-relevant state variable; thus, the solution concept has the flavor of a Markov-Perfect Equilibrium \citep{maskin2001markov}. This solution concept is standard, widely used in the prior literature and that on decentralized search and bargaining.\footnote{While we do note pursue a non-cooperative microfoundation here, one may view each bilateral bargaining game as the outcome of matched agents playing ``quick'' random proposer bargaining games that break down with some probability, as in \cite*{binmore1986nash}.} The prepay equilibrium studied in \Cref{Section-Prepay} is non-Markovian as behavior varies with past payments, even as the set of informed players remains fixed. 
\section{The Challenge of Information Resales}\label{Section-Trading}
%

Here, we focus on the set of Markovian equilibria, where we specify that trading decisions and prices condition on only the payoff-relevant state variable, $M$. For notational convenience, we redefine $\alpha: G \times \mathcal{M} \rightarrow \{ D, A_M, A_I \}$, $p: G \times \mathcal{M} \rightarrow \Re$, and $V: A \times \mathcal{M} \rightarrow \Re$, as these do not vary across histories $h, h'$ that have the same set of currently informed agents $M \equiv M(h) = M(h')$. 
Observe that in this Markovian setting, there would be no payments without the transfer of information; because such payments would not influence future behavior, \Cref{Equation-pricefunction} implies that such prices equal $0$. Therefore, the only kind of agreements that form here are those that involve the transfer of information. 

Without the prospect for resale, a single seller would obtain $wv$ from \emph{each} buyer, accumulating a total of $n_B wv$. We prove below that $wv$ is a tight upper bound to the seller's \emph{total} payoff across all Markovian equilibria. \Cref{Section-PreliminaryResults} asserts preliminary results on prices and existence of equilibria, and \Cref{Section-BoundingPricesAcrossEquilibria} proves our main results. 

\subsection{Preliminary Results}\label{Section-PreliminaryResults}

We first fix a trading function $\alpha$, and given that, we derive continuation values and prices that are conditionally consistent with Rational Expectations and Nash Bargaining. This trading function may be inconsistent with Nash Bargaining in that we do not impose \Cref{Equation-AgreementDecisions}: the trading function could be compelling pairs to trade (resp. not trade) even if their joint surplus from trading is below (resp. above) that from not trading. This intermediate step of deriving prices and continuation values for arbitrary trading decision rules is useful for our subsequent results. 

Given a trading function $\alpha$, we say that link $ij$ is \emph{active} when the set of informed agents is $M$ if $\alpha(i,j,M)=A_I$. In other words, if $i\in M$ and $j\notin M$, then when agents $i$ and $j$ have a trading opportunity, they trade information so that their behavior induces a transition of the set of informed agents from $M$ to $M\cup \{j\}$. By convention, for any active link $ij$, $i$ will be the informed agent.

\begin{proposition}
\label{Proposition-UniquePrices}
For any trading function $\alpha$, the following hold:
\begin{enumerate}[nolistsep]
\item There exist unique prices consistent with Rational Expectations and Nash Bargaining.
\item If for every proper subset $M' \subset A \setminus M$, there is more than one active link in state $M \cup M'$, then for all links $ij$ that are active in state $M$, $p(i, j, M) \rightarrow 0$ as $\Delta \rightarrow 0$. 
\item If for some state $M$, there is only one active link $ij$, then
\begin{align*}
\lim_{\Delta\rightarrow 0}	p(i, j, M) = w v + w \lim_{\Delta \rightarrow 0} V(j, M \cup \{ j \}) - (1 - w) \lim_{\Delta \rightarrow 0} V(i, M \cup \{ j \}) .
\end{align*}
\end{enumerate}
\end{proposition}

\Cref{Proposition-UniquePrices} establishes that the trading function pins down, for every period length, the prices and continuation values consistent with Rational Expectations and Nash Bargaining. Furthermore, it specifies an important consequence of competition: if there are multiple active links now and in all future states, then prices converge to $0$ as the period length vanishes. Finally, if there is a single active link, then the agents divide their joint surplus from trading according to the Nash bargaining weights.  

We use \Cref{Proposition-UniquePrices} to constructively prove that an equilibrium exists for low frictions. We consider an \emph{immediate agreement} trading function, which specifies that links are active whenever possible. We prove that this trading function generates prices and continuation values such that this trading function is also consistent with our equilibrium conditions, resulting in an immediate agreement equilibrium.
\begin{proposition}\label{prop:immediate_agreement}
There exists $\overline{\Delta}>0$ such that for all $\Delta<\overline\Delta$, there exists an immediate agreement equilibrium. 
\end{proposition}


As every link is active in this equilibrium, \Cref{Proposition-UniquePrices} implies that prices converge to $0$ as $\Delta\rightarrow 0$. Immediate agreement equilibria are the focus of prior work \citep{polanski2007decentralized,manea2021bottleneck}, which have shown that these equilibria exist generally if all agents are \emph{perfectly patient}. Because our model involves frictions, our method of proof is different. We use the simplicity of a complete graph and identical values of information to construct an equilibrium with small positive frictions. 

Our focus is not on the immediate agreement equilibrium per se, other than as a way to assure equilibrium existence; instead, we investigate whether a seller's inability to capture a large fraction of the social surplus is a conclusion that holds across all Markovian equilibria. We turn to this question next. 

\subsection{Main Results: Bounding Prices Across All Equilibria}\label{Section-BoundingPricesAcrossEquilibria}
In \Cref{Section-OptimalEqmExample}, we see that a seller may use delay credibly to obtain a non-trivial share of the surplus (in contrast to the immediate agreement equilibrium) from one of the buyers. We prove that this upper-bound is general: a monopolist can capture a non-trivial fraction of the surplus from \emph{at most} one buyer, because once there are two or more informed agents, prices converge to $0$ across all Markovian equilibria. 

\begin{proposition}\label{thm:twoinfprice}\label{Proposition-TwoInformedPlayers}
If there are at least two informed agents, then equilibrium prices converge to $0$ in the frictionless limit. Formally, for any equilibrium $(\alpha,p,V)$ and $\Infmd \in \mathcal{M}$, if $|\Infmd|\geq 2$, then for all $s\in \Infmd, \ b\in \Uninfmd$, $\lim_{\Delta \rightarrow 0} p(s, b, M) = 0 $. 
\end{proposition}
\begin{proof}
Given \Cref{Proposition-UniquePrices}, it suffices to show that for any equilibrium, if $|\Infmd|=m \geq 2$, then there must be more than one active link. The proof proceeds by induction on the number of uninformed agents. \smallskip

\noindent\textbf{Base Case:}
Consider the base case where there is a single uninformed agent, $b$, and therefore, $\Infmd=A\backslash\{b\}$. 
Suppose that buyer faces a trading opportunity with an informed agent. Because, $V(i, A) = 0$ for every $i \in A$, it follows that the joint surplus from information trade is $v$. The highest possible joint surplus without information trade is $\delta v < v$. Therefore, in every equilibrium, $\alpha(i,b,A\backslash\{b\})=A_I$ for every $i\in A\backslash\{b\}$. Hence, there are at least two active links, and so \Cref{Proposition-UniquePrices} implies that prices converge to $0$. 
\smallskip

\noindent\textbf{Inductive Step:}
Now, suppose that whenever $n-m \leq \ell$, there is more than one active link. Consider the case where $n-m = \ell + 1$. Suppose towards a contradiction that there is at most one active link. In equilibrium, it cannot be that there are zero active links because this would imply disagreement payoffs of $0$, in which case trading between any pair of informed and uninformed agents increases joint surplus. Now, suppose that there is one active link $sb$ in state $M$. It follows that 
\begin{align*}
\lim_{\Delta \rightarrow 0} p(s, b, M) = w v +  w \lim_{\Delta \rightarrow 0} V(b, M \cup \{ b \}) - (1 - w) \lim_{\Delta \rightarrow 0} V(s, M \cup \{ b \}) = w v,
\end{align*}
where the first equality follows from Part 3 of \Cref{Proposition-UniquePrices} and the second equality follows from the inductive hypothesis. 
Because there are at least two informed agents in state $M$ (by assumption), there exists an informed agent $s' \neq s$ for whom $\alpha(s', b, M) \neq A_I$. For this to be consistent with Nash Bargaining,
\begin{align*}
\delta  V(s', M \cup \{ b  \} ) + v + \delta  V(b, M \cup \{ b \}) \leq \delta  V(s', M) + \delta  V(b, M).
\end{align*}
By the inductive hypothesis, the LHS converges to $v$ as $\Delta\rightarrow 0$. Therefore, taking the limit of the above inequality, $\lim_{\Delta \rightarrow 0} [ V(s', M) + V(b, M) ] \geq v$. Moreover, $\lim_{\Delta \rightarrow 0} V(s', M) = \lim_{\Delta \rightarrow 0} V(s', M \cup \{ b \}) = 0$, because $s'$ does not trade in state $M$. Thus, $\lim_{\Delta \rightarrow 0} V(b, M) \geq v$. However, $\lim_{\Delta \rightarrow 0} V(b, M) = v - \lim_{\Delta \rightarrow 0} p(s, b, M) = (1 - w) v < v$, leading to a contradiction.
\end{proof}

\Cref{thm:twoinfprice} implies that once there are at least two possible sellers, competition forces prices to converge to $0$. This Bertrand-like outcome emerges even though all trading occurs in bilateral meetings, with small frictions from delay, and where pricing can be personalized. We note that this argument is more general than our model; indeed, it applies even if agents have heterogeneous link-specific bargaining weights as well as different values for information.\footnote{We conjecture that this result may apply also to non-Markovian equilibria. The intuition is, we think, that once there is a single buyer and multiple sellers, prices converge to $0$ across all equilibria, and therefore, by induction, prices vanish once there are $2$ or more informed agents. } 

When are prices strictly positive and bounded away from $0$? A corollary of \Cref{thm:twoinfprice} is that this can happen only if (a) an information seller is a monopolist, and (b) she has not yet sold information to any buyer. Thus, prices are strictly positive only on the first sale. The result below constructs an equilibrium that attains this upper-bound: the single seller of information, $s$, designates a single buyer as the first to whom she would sell information, disagreeing with every other buyer before she does so. Such behavior garners her a limit payoff of $w v$ and as we show below, is consistent with Nash Bargaining. Her payoff is a $1/n_B$ fraction of what she achieves were resale prohibited, and thus, when facing a large group of buyers, resale possibilities severely limit her capability to appropriate a significant fraction of social surplus. This result generalizes what we illustrated in \Cref{Section-OptimalEqmExample}. 

\begin{proposition}\label{prop:endogenous_bottleneck}
Suppose that there is a single seller of information, $s$. There exists $\overline{\Delta}>0$ such that for all $\Delta<\overline\Delta$, and for every buyer $b$, there exists an equilibrium where the seller only sells information \emph{first} to buyer $b$, disagreeing with all other buyers until she does so. After she trades information with that buyer, the equilibrium features immediate agreement. As $\Delta \rightarrow 0$, the price charged to $b$ converges to $w v$ and all other prices converge to $0$.
\end{proposition}

\begin{proof}

\autoref{prop:immediate_agreement} assures that, from any starting state and for small enough $\Delta$, there exists an equilibrium with immediate agreement. Therefore, all we need to construct is behavior in the initial state $M^0 = \{s\}$. Consider a trading function $\alpha$ such that for some buyer $b$, $\alpha(s,b,M^0)=A_I$ and for all other buyers $b'$, $\alpha(s,b',M^0)=D$. Thus,
\begin{align*}
	\lim_{\Delta\rightarrow 0} p(s,b,M^0) & = w v+ w \lim_{\Delta\rightarrow 0}V(b,M^0\cup\{b\}) - (1 - w) \lim_{\Delta\rightarrow 0} V(s,M^0\cup\{b\}) = w v,
\end{align*}
where the first equality follows from Part 3 of \cref{Proposition-UniquePrices}, and the second follows from \cref{thm:twoinfprice}. To prove that this is an equilibrium, we have to show that the trading function $\alpha$ specified above is consistent with Nash Bargaining. 

Since $sb$ is the only active link in $M^0$, the joint disagreement surplus is less than the joint agreement surplus because disagreement merely delays payoffs. Therefore, it follows that $\alpha(s, b, M^0) = A_I$ (recall that agreement without information transfer is identical to disagreement in the Markovian setting). Consider interactions between $s$ and any buyer $b' \neq b$. Disagreement is consistent with equilibrium as long as the following inequality holds:
\begin{equation}\label{Inequality-Disagreement}
\delta  V(s, M^{0} \cup \{ b' \}) + v + \delta  V(b', M^{0} \cup \{ b' \}) \leq \delta  V(s, M^{0}) + \delta V(b', M^{0}).
\end{equation}
\Cref{thm:twoinfprice} implies that as $\Delta \rightarrow 0$, $V(s, M^{0} \cup \{ b' \}) \rightarrow 0$ and $V(b', M^{0} \cup \{ b' \}) \rightarrow 0$. Because of the initial trading price, $V(s, M^{0}) \rightarrow w v$. Since $b'$ does not trade until after $b$ has become informed and thus pays nothing, $V(b', M^{0}) \rightarrow v$. Hence, as $\Delta\rightarrow 0$, the LHS converges to $v$, and the RHS converges to $(1 + w) v$, so \eqref{Inequality-Disagreement} holds for sufficiently low values of $\Delta$.  
\end{proof}
%
%
%
%
%



\section{The Prepay Solution}\label{Section-Prepay}

Having shown that a monopolistic seller appropriates very little of the social surplus from information, we turn to a non-contractual avenue for her to do better. We consider a  simple non-Markovian equilibrium, which we describe as the \emph{prepay equilibrium}, that gives the seller approximately $n_Bwv$ in the frictionless limit. This payoff coincides with that she obtains if buyers were prohibited from reselling information, and therefore completely solves the resale problem.

To summarize this equilibrium, we allow the state variable to also include past payments. A generic state is now $(M,K)$ where $M$ still refers to the set of agents with the information good and $K$ refers to the set of buyers who have given a monetary transfer to the initial seller. All history dependence is summarized by this state, so now $\alpha: G \times \mathcal{M} \times 2^{B} \rightarrow \{ D, A_M, A_I \}$, $p: G \times \mathcal{M} \times 2^{B} \rightarrow \Re$, and $V: A \times \mathcal{M} \times 2^{B} \rightarrow \Re$.

Let us describe the trading decisions of this prepay equilibrium. For any state $(M, K)$ with $M = \{ s \}$ and $|K| < n_{B}-1$,  
	\begin{align*}
\alpha(i, j, M, K)= \begin{cases}
                       A_M &\text{ if } i=s,j\in B\backslash K, \\
                        D & \text{ otherwise}.
                    \end{cases}\end{align*}
That is, trades occur only between the initial seller and buyers who have not yet prepaid, and trades involve only monetary transfers. Once $n_B-1$ prepayments are made, there is no further prepayment. The information good is then sold exclusively to the buyer who has not prepaid: for any $K$ such that $|K|=n_B-1$, 
	\begin{align*}
\alpha(i, j, \{s\}, K)= \begin{cases}
                       A_I &\text{ if } i=s,j\in B\backslash K, \\
                        D & \text{ otherwise}.
                    \end{cases}\end{align*}
If the seller faces a buyer who has already prepaid, that trading opportunity ends in disagreement. After the sale of information to the buyer $j$ who has not prepaid---and more generally for any state where information is possessed by at least two agents--- or in any state where all buyers have prepaid, we transition to the immediate agreement equilibrium where every trading opportunity between an informed and an uninformed agent results in information trade: if $M\neq \{s\}$ or $|K| = n_{B}$,
\begin{align*}
\alpha(i, j, M, K)= \begin{cases}
                       A_I &\text{ if } i\in M,j\notin M, \\
                        D & \text{ otherwise}.
                    \end{cases}\end{align*}
This fully describes the trading decisions. Note that monetary transfers ($\alpha = A_M$) between two buyers or between two sellers does not change the state, so \Cref{Equation-AgreementDecisions} implies that this is never jointly better than disagreement, and \Cref{Equation-pricefunction} shows that this must be a zero transfer, anyway.
The following result proves existence of this prepay equilibrium, its consistency with Nash Bargaining, and derives the seller's payoff in the frequent-offer limit. 
\begin{proposition}\label{prop:prepay}
Suppose that there is a single seller of information, $s$. There exists $\overline{\Delta}>0$ such that for all $\Delta<\overline\Delta$, the prepay equilibrium exists. Moreover, as $\Delta \rightarrow 0$, the prices paid to the initial seller in the prepay equilibrium all converge to $w v$, while all other prices converge to $0$.
\end{proposition}

The proof of \Cref{prop:prepay} is in the Appendix. Its logic combines that of Nash Bargaining for a scarce good with that of \Cref{prop:endogenous_bottleneck}. Once all prepayments are made, the remaining buyer pays the bilateral bargaining price for information ($\approx wv$) for the reasons described in \Cref{prop:endogenous_bottleneck}, since the continuation behavior is then identical to that of \Cref{Section-BoundingPricesAcrossEquilibria}. After this buyer buys information, competition in the resale market ensures that the price of information charged to all other buyers converges to $0$, as outlined in \Cref{Proposition-TwoInformedPlayers,prop:endogenous_bottleneck}. 
Prepaying therefore puts a buyer in the position to buy information in the (near-)future at a price of $\approx 0$, generating a price discount of approximately $wv$ relative to the buyer who does not prepay; thus, the equilibrium value of prepayment is approximately $wv$. We now argue that because only $n_B-1$ buyers can prepay, the price of prepayment converges to $wv$. 

To see why, suppose the seller meets a buyer who has not yet prepaid, but there still is an opportunity to prepay. On the buyer's end, the gain from prepaying now is that he can guarantee that he can purchase information for $\approx 0$; if he delays, there is a chance that the next time he meets the seller, all other buyers prepay and he is left to pay $wv$ for information. By contrast, for the seller, the only gain from collecting a prepayment now versus waiting is the cost of delay, which vanishes in the frequent-offer limit. For these gains to remain magnitudes of the same proportional size, it must be that in the frequent-offer limit, the buyer's gain from prepayment also vanishes, and hence, the prepayment converges to $wv$. 

The above outlines the on-path equilibrium logic for how the seller extracts $wv$ from each buyer. She must also not sell information too early (before she has collected all of the prepayments) nor sell information initially to a buyer who has already prepaid. We show in the proof that neither issue arises.\footnote{The logic in each case is straightforward. Selling the information good early (before $|K| = n_{B}-1$) is inconsistent with Nash bargaining, because the market for prepayment would shut down, eliminating a source of extra surplus for the monopolist seller. Moreover, selling the information good initially to a buyer who has already prepaid is inconsistent with Nash bargaining, because that buyer would not be willing to pay anything in the frictionless limit, as they could simply wait until a buyer who has not prepaid (or will not prepay) trades for the information good.}

Let us interpret this prepay equilibrium. Observe that the seller obtains the same value that she would if there were full intellectual property protection, without the buyers committing not to resell information. Equivalently, this payoff coincides with the seller's payoff if the trading network were a star centered on the seller, which Theorem 1 of \cite{polanski2007decentralized} shows is the seller-optimal network. A strategic logic from the prior literature \citep{polanski2007decentralized,manea2021bottleneck} is that \emph{bottlenecks} in Markovian equilibria allow the seller to capture surplus; what benefits the seller is that for each buyer, there is a single path along which information can pass from the initial seller. Our prepay equilibrium effectively creates these bottlenecks: since prepayment can be done only to the seller, it is non-replicable, and the buyer who does not do any prepayment effectively faces a bottleneck for information. 

Our construction solves the commitment problem by introducing a modicum of history dependence. We view the history dependence here to be distinct from that of repeated games---we note that the setting here is not that of a pure repeated game---or analogous reputational constructions \citep{ausubel1989reputation} in bargaining environments. The role of history dependence is not to punish the seller
because she does not have a myopic incentive to deviate from the prepay equilibrium. Rather, each buyer may have a myopic incentive not to prepay, and the prepayment scheme addresses it by leveraging market competition: if he chooses not to prepay, a buyer has to pay the full price for the information good, whereas by prepaying, he obtains information cheaply on the resale market. We observe that equilibrium behavior in this prepayment scheme is not being sustained by the threat of triggering off-path punishments.
Thus, one way to see this result is that this mild form of history-dependence allows the seller to not merely sidestep the resale issue but to instead exploit it.\footnote{This raises the question whether other equilibria could do even better than the commitment benchmark of full intellectual property protection. Answering this question is challenging because it would require a full non-cooperative treatment. As our interest is in identifying how intellectual property protection is obviated by a mild form of history dependence, we view this intriguing question to be beyond the scope of this study.} One may also wonder whether creating a second iteration of prepayment---allowing buyers to prepay for the opportunity to prepay for information---results in even higher payoffs for the seller. Our intuition is that this is not so: as the gains from prepaying versus buying information vanish as $\Delta\rightarrow 0$, buyers would not be willing to pay more than an amount vanishing in $\Delta$ for the right to prepay.

We view the logic of this prepay equilibrium to apply beyond this model. So as to construct an immediate agreement equilibrium, we assume that bargaining weights do not vary across the network and that the value of information is identical across agents; that said, we conjecture that an immediate agreement equilibrium exists more generally. Whenever an immediate agreement equilibrium exists, our main results follow: \Cref{Proposition-TwoInformedPlayers} implies that the price converges to $0$ once there are two or more informed parties in any Markovian equilibrium, \Cref{prop:endogenous_bottleneck} would imply that the seller can appropriate surplus from at most one buyer in a Markovian equilibrium, and \Cref{prop:prepay} would imply that the seller can obtain the commitment benchmark by using a prepayment scheme.

One question about the realism of this solution is that buyers are assumed to  observe the full history, including how many other buyers have prepaid. In some contexts, we believe that blockchains can be used to alleviate this problem: the information good seller would reveal their public address, and the first $n_{B} - 1$ buyers would send a one-way transfer of $w v$ to the seller, publicly recorded on the blockchain. Each buyer (and the seller) can tell if they are one of the first $n_{B} - 1$ buyers by first checking the blockchain for transactions involving that particular seller public address. Then, the $n_{B}$'th buyer can initiate a transaction with the seller, paying $w v$ in exchange for the information good. Finally, the information good is then exchanged freely for price $0$ in any fashion, i.e., these do not need to be blockchain transactions.\footnote{Alternatively, smart contracts could be used as a part of the prepay equilibrium to automate the information good release process, as triggered by the transfers.} In this implementation, the seller cannot benefit from deviating to trade the information good offline (i.e., in a way that is not recorded on the blockchain) nor from faking additional transactions (e.g., by posing as a legitimate buyer), because in the prepay equilibrium, each buyer expects to pay exactly $w v$, regardless of how many other buyers have already transacted.\footnote{One might wonder if smart contracts, which use blockchains to automatically implement contracts without interfacing with the legal system,  provide a straightforward contractual solution to the information good trading problem by forbidding resale. Once a buyer has the information good, there is nothing to stop him from engaging in offline transactions with that information good, and \Cref{Proposition-TwoInformedPlayers} still applies. These transactions are almost certainly not observable to the smart contract, which prevents penalty clauses from being executed.}

\section{Conclusion}

This paper studies how the possibility for resale influences the pricing of information. Prices converge to $0$ across Markovian equilibria as soon as two agents possess information. In such equilibria, a monopolistic seller obtains at most a share of one buyer's value of information. Commitment problems inherent in replicability and resale impede the monopolistic seller to do any better even if she has substantial market power, can personalize prices, and exploit (slight) search frictions. 

These commitment problems can be used to craft an antidote: there is a (non-Markovian) prepay equilibrium where the seller's payoff is approximately equal to that of perfect intellectual property protection. Prepaying allows a buyer to purchase information at a negligible price, once it has been replicated and resold. But by allowing only $n_B-1$ buyers to prepay, the seller appropriates each buyer's value of prepayment. 

\appendix

\section{Appendix}\label{app:Proofs}

\begin{proof}[{Proof of \Cref{Proposition-UniquePrices} on page \pageref{Proposition-UniquePrices}}]

First, note that the existence of unique prices for inactive links is trivial. The state does not change after bargaining occurs on these links. Therefore, \Cref{Equation-pricefunction} is solved only by prices equal to $0$. The rest of the proof considers only active links.

For a trading decision function $\alpha$, let
\begin{align*}
B(i, M) &\equiv \{ j \in A \setminus M : \alpha(i, j, M) = A_I \},\\
S(j, M) &\equiv \{ i \in M : \alpha(i, j, M) = A_I \}.	
\end{align*}
When the set of informed agents is $M$, $B(i,M)$ is the set of buyers who trade information with seller $i$, and $S(j,M)$ is the set of sellers who trade information with buyer $j$. 
Let $\mathcal{L}(M) \equiv \{ ij | \alpha(i, j, M) = A_I \}$ denote the set of active links and $\mathcal{L}^{c}(M)$ denote its complement. Finally, let $\hat{\rho}(\delta,M) \equiv \frac{\rho}{1 - \delta  \rho  | \mathcal{L}^{c}(M) |}$, where recall that $\rho\equiv \frac{2}{n(n-1)}$ is the probability that a specific link is recognized. To economize on notation, we suppress the arguments of $\hat\rho$ but note that in the frequent offer limit, $\hat\rho(\delta,M)\rightarrow \frac{1}{|\mathcal{L}(M)|}$. 

Our proof follows by induction. \bigskip

\noindent\underline{Base Case:} Suppose that $|A \setminus M|=1$, with the remaining uninformed buyer being agent $j$. For any $i \in S(j, M)$, Nash Bargaining implies that 
\begin{align*}
& \frac{(1 - w)}{w}  \underbrace{p(i, j, M) \left( 1 - \delta  \hat{\rho} \right) }_{\text{Seller }i's\text{ gain from trade}}  =  \underbrace{v - p(i, j, M) - \delta  \hat{\rho}  \left( | \mathcal{L}(M) |  v - \sum_{i' \in S(j, M)} p(i', j, M) \right)}_{\text{Buyer $j$'s gain from trade}}, 
\end{align*}
where recall that $w$ is the seller's bargaining weight. This equation can be re-written as 
\begin{align*}
p(i, j, M) - \frac{w  \delta \hat{\rho}}{1 - \delta  \hat{\rho}}  \sum_{i' \in S(j, M) \setminus \{ i \}} p(i', j, M) = \frac{w  (1 - \delta  \hat{\rho}  | \mathcal{L}(M) |)}{1 - \delta  \hat{\rho}}  v.
\end{align*}
An analogous equation holds for any other $i' \in S(j, M)$, and subtracting that equation from the one above implies that $p(i, j, M) = p(i', j, M)$, so prices are symmetric. Therefore, for every $ij \in \mathcal{L}(M)$,
\begin{align*}
p(i, j, M) = \frac{w (1 - \delta  \hat{\rho} | \mathcal{L}(M) |)}{1 - \delta  \hat{\rho}  ( w | \mathcal{L}(M) | + (1 - w) )}  v,
\end{align*}
which proves that a unique solution exists. Note that if $|\mathcal{L}(M)|=1$, then $p(i,j,M)=wv$. If $\mathcal{L}(M)>1$, then $\lim_{\Delta\rightarrow 0} p(i,j,M)=0$.\footnote{As $\Delta\rightarrow 0$, $\delta  \hat{\rho} | \mathcal{L}(M) | \rightarrow 1$, which implies that the numerator converges to $0$. The denominator converges to $(1-w)(1-(|L(M)|)^{-1})$, which is strictly positive for $|L(M)|\geq 2$.} 

\bigskip

\noindent\underline{Inductive Step:} Suppose that for every $M'$ with $|M'| \geq x$, prices $p(i, j, M')$ are uniquely determined for every $ij \in \mathcal{L}(M')$, and if $|\mathcal{L}(M')|\geq 2$ these prices all converge to $0$ as $\Delta \rightarrow 0$. Prices being uniquely determined implies that for each agent $k$, $V(k, M')$ is also unique (by Rational Expectations and the inductive hypothesis). We first argue that prices in $M$ with $|M|=x-1$ must also exist and be unique. 

For any $ij \in \mathcal{L}(M)$ (where we adopt the convention that the first agent listed in the link is always the informed agent: $i \in M$), Nash Bargaining implies that
\begin{multline}\label{Equation-RNBUnique}
 (1 - w)  \left[ p(i, j, M) + \delta  V(i, M\cup\{j\}) - \delta  V(i, M) \right] \\
 = w  \left[ v - p(i, j, M) + \delta  V(j, M\cup\{j\}) - \delta  V(j, M) \right].
\end{multline}
%
%
%
%
Using Rational Expectations to expand the value functions in \eqref{Equation-RNBUnique}, we obtain 
\begin{align*}
& (1 - w)  \left[ p(i, j, M) + \delta  V(i, M\cup\{j\}) - \delta  \hat{\rho}  \left[ \sum_{j' \in B(i, M)} p(i, j', M) + \delta  \sum_{i'j' \in \mathcal{L}(M)} V(i, M\cup\{j'\}) \right] \right] \\
& = w  \Biggl[ v - p(i, j, M) + \delta  V(j, M\cup\{j\})  - \delta  \hat{\rho}  \left[ \sum_{i' \in S(j, M)} [v - p(i', j, M)] + \delta  \sum_{i'j' \in \mathcal{L}(M)} V(j, M\cup\{j'\}) \right] \Biggr].
\end{align*}
We collect all terms that involve prices in state $M$ on the LHS and others on the RHS. Define 
\begin{align*}
\kappa(i, j, M) & \equiv w  (1 - \delta  \hat{\rho}  | S(j, M) |)  v +  w  \delta  \left[ V(j, M\cup\{j\}) - \hat{\rho}  \delta  \sum_{i'j' \in \mathcal{L}(M)} V(j, M\cup\{j'\}) \right] \\
& - (1 - w)  \delta  \left[ V(i, M\cup\{j\}) - \hat{\rho}  \delta  \sum_{i'j' \in \mathcal{L}(M)} V(i, M\cup\{j'\}) \right].
\end{align*}
Then it follows that
\begin{align*}
 & [1 - \delta  \hat{\rho}]  p(i, j, M) - (1 - w)  \delta  \hat{\rho}  \sum_{j' \in  B(i, M) \setminus \{ j \}} p(i, j', M) - w  \delta  \hat{\rho}  \sum_{i' \in S(j, M) \setminus \{ i \}} p(i', j, M) = \kappa(i, j, M). 
\end{align*}
Note that $\kappa(i, j, M)$ is uniquely determined in equilibrium since it depends only on parameters and future continuation values, which are uniquely determined. Similar equations hold for every active link between a buyer and a seller in state $M$. 

We index these links as $1,\ldots, | \mathcal{L}(M) |$ (the specific indexing is unimportant) and define an $| \mathcal{L}(M) | \times | \mathcal{L}(M) |$ matrix $\Phi(M)$ where $\Phi_{uv}(M)=w$ if link $u$ and link $v$ are distinct active links and share a common buyer,  $\Phi_{uv}(M)= 1 - w$ if link $u$ and link $v$ are distinct active links and share a common seller, and $\Phi_{uv}(M) = 0$ otherwise. Combining the $| \mathcal{L}(M) |$ equations yields the following matrix equation:
\begin{align}
\label{eqn:priceeqn}
& \left[ I_{| \mathcal{L}(M) |} - \frac{\delta  \hat{\rho}}{1 - \delta  \hat{\rho}}  \Phi({\Infmd}) \right] \vec{p}(M) = \frac{1}{1 - \delta  \hat{\rho}}  \vec{\kappa}(M)
\end{align}
Here, $\vec{p}(M)$ and $\vec{\kappa}(M)$ are $| \mathcal{L}(M) | \times 1$ vectors consisting of all of the state $M$ prices and values of $\kappa$. Note that the prices for the active links form a self-contained system, as activity on other links does not change the state. Since the right hand side of this last equation is unique, the current state price vector $\vec{p}(M)$ is unique if matrix $\Psi(M) \equiv \left[ I_{| \mathcal{L}(M) |} - \frac{\delta  \hat{\rho}}{1 - \delta  \hat{\rho}}  \Phi({\Infmd}) \right]$ is invertible.

Note that $\Psi({\Infmd})$ is a Z-matrix because the off-diagonal elements are all non-positive. Additionally, we can show that $\Psi(M)$ exhibits semipositivity; that is, there exists a vector $\vec{x}>0$ such that $\Psi(M)\vec{x}>0$.\footnote{This is shown by using for $\vec{x}$ a vector of all ones and noting that for all $u$, 
\begin{align*}
[\Psi(M)\vec{x}]_{u} & = 1 - [ (1 - w) | B(i, M) | + w | S(j, M) | - 1 ] \frac{\delta  \hat{\rho}}{1 - \delta  \hat{\rho}} > 1 - (| \mathcal{L}(M) | - 1 )  \frac{\frac{\rho}{1 - \rho  | \mathcal{L}^{c}(M) |}}{1 - \frac{\rho}{1 - \rho  | \mathcal{L}^{c}(M) |}} \\
& = 1 - \frac{\rho  ( | \mathcal{L}(M) | - 1)}{1 - \rho  [ | \mathcal{L}^{c}(M) | + 1 ]} = 0.
\end{align*}}
Being a Z-matrix that exhibits semipositivity is equivalent to $\Psi(M)$ being a non-singular M-matrix and thus invertible \citep{plemmons1977m}. This completes the proof of existence and uniqueness of the price vector.

To show that prices converge to $0$ whenever more than one link is active in the current and all subsequent states (the $M \cup M'$ in the statement of the proposition), consider the values of $\kappa(i, j, M)$. From the inductive hypothesis, every informed agent's continuation payoff on a subsequent state converges to $0$ (they capture nothing from resale). Moreover, every buyer's continuation payoff on a subsequent state converges to $v$ (they pay nothing to buy, but they capture nothing from resale). Also, note that $\hat{\rho} \rightarrow \frac{1}{| \mathcal{L}(M)|}$. This implies that
\begin{align*}
\kappa(i, j, M) & \rightarrow w  \left( 1 - \frac{| S(j, M) |}{| \mathcal{L}(M) |} \right)  v -  w \left[ \frac{1}{| \mathcal{L}(M) |}  \sum_{i'j' \in \mathcal{L}(M), j' \neq j} v \right] = 0
\end{align*}
Since the solution price vector is unique, and the right hand side of the price equation $\rightarrow 0$ (when $| \mathcal{L}(M) | > 1$), the solution price vector must also $\rightarrow 0$. 

When $| \mathcal{L}(M) | = 1$, the right hand side of \Cref{eqn:priceeqn} takes on an indeterminate form in the limit. In this case, $I_{| \mathcal{L}(M) |} = 1$. By the definition of matrix $\Phi(M)$, the only non-zero entries correspond to pairs of distinct active links, of which there are none, so $\Phi(M) = 0$. Substituting in the expression for $\kappa(i, j, M)$ and taking the limit $\Delta\rightarrow 0$, 
\begin{align*}
\lim_{\Delta \rightarrow 0} p(i, j, M) &= \lim_{\Delta \rightarrow 0} \frac{\kappa(i, j, M)}{1 - \delta \hat{\rho}} =  w v + w \lim_{\Delta \rightarrow 0} V(j, M \cup \{ j \}) - (1 - w) \lim_{\Delta \rightarrow 0} V(i, M \cup \{ j \}) 
\end{align*}
\end{proof}

\begin{proof}[{Proof of \cref{prop:immediate_agreement} on page \pageref{prop:immediate_agreement}}]

We define ``immediate agreement'' as $\alpha(i, j, M) = A_I$ for all $(i, j, M)$ for which either $i \in M, j \notin M$ or $i \notin M, j \in M$ (although we adopt the convention that $i$ is the informed agent).
We first prove that for every $\Delta$, immediate agreement prices are symmetric: $\forall M$ and $\forall i \in M, \forall j \in A \setminus M, p(i, j, M) = p(m)$, i.e., prices depend only on the size of the informed set. We establish this claim by induction. The base case ($m = n - 1$) symmetry was already proven for  \Cref{Proposition-UniquePrices}. \smallskip

\noindent\underline{Inductive Step:} Suppose that for every $M'$ with $|M'| \geq x$, prices $p(i, j, M')$ are symmetric. Since all seller/buyer links are active, we write the price equation by indexing links as $1,\ldots, [m  (n - m)]$. We use
\begin{align}\label{Equation-RedundantLinks}
    R(m)\equiv\frac{m (m - 1) + (n - m) (n - m - 1)}{2}
\end{align}
to denote the number of redundant links---between two informed agents or two uninformed agents---when the number of informed agents is $m$. $\Phi(M)$ and $\Psi(M)$ are defined in the same way as in the proof of \Cref{Proposition-UniquePrices}, so the price equation is
\begin{align*}
& \Psi(M) \vec{p}(M) = \frac{1}{1 - \delta  \hat{\rho}(\delta,m)}  \vec{\kappa}(M),
\end{align*}
where $\hat{\rho}(\delta,m) \equiv \frac{\rho}{1 - \delta \rho R(m)}$.

To see that prices are symmetric, first note that price symmetry in the inductive hypothesis implies that all future buyer continuation payoffs are the same and all future seller continuation payoffs are the same, and these depend only on the number of sellers. Hence, $\kappa(i, j, M)$ does not depend on the identities of the particular buyer or seller or the precise configuration, so $\vec{\kappa}(M)$ can be written as $\vec{\kappa}(M) = \kappa(m)  \vec{1}$, where $\kappa(m)$ is some scalar, and $\vec{1}$ is the vector of all ones. With symmetric prices, there exists a scalar $p(M)$ such that $\vec{p}(M) = p(M) \vec{1}$. Then, $\Psi(M) \vec{p}(M) = p(M)  \Psi(M) \vec{1}$. Because of the symmetry of the network (in each state, every buyer is linked to the same number of sellers, and every seller is linked to the same number of buyers), every row of $\Psi(M)$ has the same sum. Moreover, this row sum depends only on the number of sellers. Hence, $\Psi(M) \vec{1} = s(m)  \vec{1}$, where $s(m)$ is the (scalar) row sum. Now, the price equation can be written as $p(M) s(m)  \vec{1} = \frac{1}{1 - \delta  \rho(m)}  \kappa(m)  \vec{1} \Leftrightarrow p(M) = \frac{1}{1 - \delta  \rho(m)} \frac{\kappa(m)}{s(m)} \equiv p(m)$. Because the pricing equation determines unique prices, $p(M)$ is this unique, symmetric price. 

We use this symmetry result to write state $M$ prices as $p(m)$, buyer continuation payoffs as $V^{b}(m)$, and seller continuation payoffs as $V^{s}(m)$.\footnote{Also, recall from the proof of  \Cref{Proposition-UniquePrices} that prices are $0$ on all inactive links.}
\begin{align}
V^{s}(m) & = \hat{\rho}(\delta,m)  \left[ (n-m) p(m) +  \delta m  (n-m) V^{s}(m+1) \right],\label{Equation-SellerContValue} \\
V^{b}(m) & = \hat{\rho}(\delta,m)  \left[ m [v - p(m)] +  \delta m  (n-m-1) V^{b}(m+1) + \delta m V^{s}(m+1) \right],\\
p(m) & = w (v+\delta V^{s}(m+1)-\delta V^{b}(m)) - (1 - w) (\delta V^{s}(m+1)-\delta V^{s}(m)).
\end{align}

%
We use the following lemma, which we prove at the end of the Appendix.  

\begin{lemma}\label{lem:positiveprices}
Suppose that trading decisions are immediate agreement. Then if $\delta$ is sufficiently high ($\Delta$ is sufficiently low), the Nash Bargaining prices on active links ($p(m)$) are positive for all $m$.
\end{lemma}

Using \cref{lem:positiveprices}, to show that an immediate agreement equilibrium exists for high $\delta$, it suffices to show that the joint surplus from agreement with information trade exceeds the joint surplus from disagreement (which equals the joint surplus from agreement without information trade), for all $m$. We show below that the equilibrium offers the seller a strictly positive gain from selling information, and as Nash Bargaining offers the seller a share of the total gain from selling information, this assures that the joint surplus from agreement with information trade exceeds that from disagreement.

Observe that \Cref{Equation-SellerContValue} can be rewritten as
\begin{align*}
    \frac{1}{\hat{\rho}(\delta,m) (n - m)} V^{s}(m) - m \delta V^{s}(m+1) = p(m).
\end{align*}
Observe that $\hat\rho(\delta,m)<\hat\rho(1,m)$, and that by algebra, $\hat\rho(1,m)(n-m)=1/m$. Therefore, we obtain
\begin{align*}
    m [ V^{s}(m) - \delta V^{s}(m+1) ] < p(m).
\end{align*}
Since $V^s(m)\geq 0$ (as prices are non-negative), multiplying both sides by $\delta$ yields that
\begin{align}\label{Equation-mdelta}
    m\delta [ V^{s}(m) - V^{s}(m+1) ] < \delta p(m) < p(m).
\end{align}
Observe that this implies that 
\begin{align}\label{Equation-1delta}
    \delta [ V^{s}(m) - V^{s}(m+1) ] < p(m).
\end{align}
To see why, note that it suffices to consider the case where the LHS of \eqref{Equation-1delta} is non-negative, as \Cref{lem:positiveprices} already assures that $p(m)\geq 0$. In that case, $m\delta [ V^{s}(m) - V^{s}(m+1) ]\geq \delta [ V^{s}(m) - V^{s}(m+1) ]$, guaranteeing that \eqref{Equation-1delta} holds. Rewriting this implies that 
\begin{align*}
    \delta V^s(m) <p(m)+\delta V^s(m+1),
\end{align*}
and hence, the seller is strictly gaining from selling information. As Nash Bargaining gives the seller a share of the total gains from trade, those total gains from trade must also be strictly positive, completing our argument. \end{proof}

\begin{proof}[{Proof of \cref{prop:prepay} on page \pageref{prop:prepay}}]

The equilibrium conditions are satisfied when $m > 1$ and in all off-equilibrium path histories, as the immediate agreement equilibrium always exists (\Cref{prop:immediate_agreement}). In both of these cases, note that purely monetary transfers could hypothetically continue to occur, but afterwards there is immediate agreement regardless, so there are no joint gains from these transfers and $\alpha(i, j, M, K) \neq A_M$ is consistent with Nash Bargaining. We proceed by induction on $| K |$.

As the base case, suppose $M = \{ s \}$ and $|K| = n_{B} - 1$, so that all prepayments have been collected but the information good has not yet been sold. This is analogous to the setting of \Cref{prop:endogenous_bottleneck}, so everything is consistent with equilibrium and the price converges to $wv$. As one small detail, note that $b \in B \setminus K$ will not give a monetary transfer ($\alpha = A_M$), because this only nominally changes the state relative to information trade ($\alpha = A_I$); the play going forward is immediate agreement either way, and the transfer simply delays information trade. Furthermore, the prices and continuation payoffs are trivially symmetric. That is, they depend only on the number of prepayments made thus far and in the case of a buyer's continuation payoff, whether or not the buyer has already prepaid.

%
%

As an inductive hypothesis, suppose that for all $(M, K)$ with $M = \{ s \}$ and $|K| = x + 1$ for some $x \in \{ 0, 1, \ldots, n_{B} - 2 \}$, the equilibrium conditions are met (for low $\Delta$), limit prices are as in the statement of the Proposition, and prices and continuation payoffs are symmetric as previously described. Suppose that $M = \{ s \}$ and $|K| = x$. 

First, consider the possibility of monetary transfers involving buyers $j \in K$. If the transfer were to take place, the state would not change, and thus the joint surplus from trade and the joint surplus from disagreement are equal. Thus, the equilibrium conditions are satisfied by $\alpha(i, j, M, K) \neq A_M$. The same argument also justifies no monetary transfers to $i \neq s$ from $j \in B \setminus K$. Since the history does not record this buyer as having transferred money to $s$, the state does not advance, so there are no joint gains from trade. We will later show that no information trade is also consistent with equilibrium.

Next, we consider monetary transfers from buyers $j \in B \setminus K$ to seller $s$. We need to redefine what is meant by an active link, which differs for the setting with monetary transfers in equilibrium. Let $\mathcal{L}(M, K) \equiv \{ ij | \alpha(i, j, M, K) \neq D \}$ denote the set of active links and $\mathcal{L}^{c}(M, K)$ denote its complement. We then define $\hat{\rho}$ analogously to the definition in the proof of \Cref{Proposition-UniquePrices}, noting that the active links in the states currently being considered are those between $s$ and buyers $j \in B \setminus K$. The price is 
\begin{align*}
 p(s, j, M, K) &= w\delta \left[  V(j, M, K \cup \{ j \}) -  V(j, M, K ) \right] - (1 - w)\delta \left[ V(s, M, K \cup \{ j \}) -  V(s, M, K ) \right],
 \end{align*}
 which can be re-written as
 \begin{align*}
 p(s, j, M, K) =& w\delta \left(  V(j, M, K \cup \{ j \}) -  \hat{\rho} \left[ - p(s, j, M, K) + \sum_{j' \in B \setminus K} \delta V(j, M, K \cup \{ j' \}) \right] \right) \\
& - (1 - w)\delta \left(  V(s, M, K \cup \{ j \})  -  \hat{\rho} \sum_{j' \in B \setminus K} [ p(s, j', M, K) + \delta V(s, M, K \cup \{ j' \}) ] \right).
\end{align*}
Re-arranging terms yields
\begin{align}
& (1 - \delta \hat{\rho}) p(s, j, M, K) - \delta \hat{\rho} (1 - w) \sum_{j' \in B \setminus ( K \cup \{ j \})} p(s, j', M, K) \nonumber \\& = w \Biggl[ (1 - \delta \hat{\rho}) \delta V(j, M, K \cup \{ j \}) - \delta \hat{\rho} \sum_{j' \in B \setminus ( K \cup \{ j \})} \delta V(j, M, K \cup \{ j' \}) \Biggr] \nonumber \\
& - (1 - w) \Biggl[ (1 - \delta \hat{\rho}) \delta V(s, M, K \cup \{ j \}) - \delta \hat{\rho} \sum_{j' \in B \setminus ( K \cup \{ j \})} \delta V(s, M, K \cup \{ j' \}) \Biggr] \label{eqn:tokenprice} 
\end{align}
We can write this more compactly as
\begin{align*}
\Psi(M, K) \vec{p}(M, K) = \frac{1}{1 - \delta \hat{\rho}} \vec{\kappa}(M, K) 
\end{align*}
where $\Psi(M, K)$ is the $[ n_{B} - | K | ] \times [ n_{B} - | K | ]$ matrix with $1$ along the diagonal and $-\frac{\delta \hat{\rho}}{1 - \delta \hat{\rho}} (1 - w)$ everywhere else, $\vec{p}(M, K)$ is the $[ n_{B} - | K | ]$-vector of prices in the current state, and $\vec{\kappa}(M, K)$ is the $[ n_{B} - | K | ]$-vector of entries as in the right hand side of \Cref{eqn:tokenprice}. $\Psi(M)$ is a non-singular M-matrix (see the proof of \Cref{Proposition-UniquePrices} for an analogous proof). The inductive hypothesis implies that the Nash bargaining prices in the current state exist and are unique; moreover, that $\frac{1}{1 - \delta \hat{\rho}} \vec{\kappa}(M, K) \rightarrow w v - (1 - w) w v$. This implies that as $\Delta \rightarrow 0$, $p(s, j, M, K) \rightarrow w v$ for all $j  \in B \setminus K$.

Also, note that each element of $\vec{\kappa}(M, K)$ is the same by the symmetry part of the inductive hypothesis, and since the matrix $\Psi(M, K)$ has a constant row sum that depends only on $|K|$, the prices only depend on $|K|$. We will write these prices as $p(|K|)$. Moreover, since all current continuation payoffs depend only on $p(|K|)$ and future continuation payoffs (which by the inductive hypothesis only depend on $|K|$), the current continuation payoffs depend only on $|K|$, whether the agent is the seller $s$, and whether the agent is a buyer who already prepaid or not. Continuation payoffs will be denoted $V^{s}(|K|)$ for the seller, $V^{b+}(|K|)$ for a buyer who has already prepaid, and $V^{b-}(|K|)$ for a buyer who has not yet prepaid. Note that the discussion of this paragraph has been all conditional on $M = \{ s \}$, and thus $M$ is suppressed in the notation; different sets $M$ give different continuation payoffs.

We now consider the continuation payoffs at the prepayment bargaining stage and that prepayment occurs only if the joint surplus from prepayment is weakly more than that of disagreement:
\begin{align*}
 V^{s}(|K|) - \delta V^{s}(|K|+1) =& \hat{\rho} (n_{B} - |K|) [ p(|K|) + \delta V^{s}(|K|+1) ] - \delta V^{s}(|K|+1)  \\=& \hat{\rho} (n_{B} - |K|) p(|K|) - [1 - \hat{\rho} (n_{B} - |K|) ] \delta V^{s}(|K|+1) \\\leq &\hat{\rho} (n_{B} - |K|) p(|K|) \end{align*}
 which implies that $\delta [ V^{s}(|K|) - V^{s}(|K|+1) ] \leq p(|K|)$. This is equivalent to $ \delta [ V^{s}(|K|) - V^{s}(|K|+1) ] \leq w (\delta V^{b+}(|K|+1) - \delta V^{b-}(|K|)) - (1 - w) (\delta V^{s}(|K|+1)-\delta V^{s}(|K|))$, which is equivalent to $w \delta [ V^{s}(|K|) - V^{s}(|K|+1) ] + w \delta [V^{b-}(|K|) - V^{b+}(|K|+1)] \leq 0 $, and therefore, $\delta [ V^{b+}(|K|+1) + V^{s}(|K|+1) ] \geq \delta [ V^{b-}(|K|) + V^{s}(|K|) ]$, i.e., prepayment is jointly better than disagreement.\footnote{Note that the above bounding argument assumed that $p(|K|) \geq 0$, which must be true for sufficiently small $\Delta$, as prices are continuous as a function of $\Delta$.} Next, we show that disagreement is jointly better than information trade in these states.

Consider potential trade of the information good in this same inductive step (that is, when $|K| < n_{B} - 1$). For any buyer $j$, we will show that the following inequality holds for small enough $\Delta$:
\begin{align*}
& v + \delta  V(s, M \cup \{ j \}, K) + \delta  V(b, M \cup \{ j \}, K ) \leq \delta  V(s, M, K) + \delta V(j, M, K)
\end{align*}

Since selling the information good to $j$ triggers immediate agreement in the future, both $V(j, M \cup \{ j \}, K )$ and $V(s, M \cup \{ j \}, K ) \rightarrow 0$. Since the seller has at least one more prepayment to collect, $\lim_{\Delta \rightarrow 0} V(s, M, K) > wv$. At worst, the buyer will still have to pay in order to get the information good, so $\lim_{\Delta \rightarrow 0} V(j, M, K) \geq (1-w)v$. Therefore, for sufficiently low $\Delta$, the inequality holds. 

Combining these two inequalities (one for purely monetary transfers, and one for information good trades), it is clear for this pair that the purely monetary transfer is at least as good as trading nothing, which is at least as good as trading the information good. Thus, prepayments with no information good trade is consistent with equilibrium (for low enough $\Delta$). This completes the inductive step, as we have demonstrated the correct limit prices, price and continuation payoff symmetry, and trading functions that are consistent with equilibrium (for low $\Delta$).
\end{proof}

\begin{proof}[Proof of \cref{lem:positiveprices} on p. \pageref{lem:positiveprices}]
Consider all markets of at least 3 agents ($n\geq 3$), because the case with two agents is trivial. The proof proceeds by induction over the number of informed agents

\noindent\underline{Base Case:} If $m = n - 1$ then by adapting the base case pricing formula from the proof of \Cref{Proposition-UniquePrices}, we get
\begin{align*}
p(m) = \frac{w [ 1 - \delta \hat{\rho}(\delta,m) (n - 1) ]}{1 - \delta \hat{\rho}(\delta,m) [ w (n - 1) + (1 - w) ]} v > 0
\end{align*}
Note that $p(m) \rightarrow 0$ as $\delta \rightarrow 1$, because $\delta \hat{\rho}(\delta,m) = \frac{\delta \rho}{1 - \delta \rho R(m)}$ converges to the reciprocal of the number of buyer/seller links (here, $n-1$). More specifically, $p(m) \in O(1 - \delta \rho(\delta, m^{*}) m^{*} (n - m^{*}))$ as $\delta \rightarrow 1$, where $m^{*}$ maximizes $[ 1 - \delta \hat{\rho}(\delta,m) m (n - m) ]$. For simplicity of notation, define $g(\delta) \equiv [ 1 - \delta \rho(\delta, m^{*}) m^{*} (n - m^{*}) ]$. Thus, $p(m) \in O(g(\delta))$ and $V^{s}(m) \in O(g(\delta))$. Furthermore, $p(m) \notin O(g(\delta)^{2})$ and $V^{s}(m) \notin O(g(\delta)^{2})$, i.e., the order of the upper bound is tight.

\noindent\underline{Inductive Step:} Given that $p(m+1) > 0$ (or $p(m+1) = 0$ in the $w = 0$ case), $p(m+1) \in O(g(\delta))$ and $V^{s}(m+1) \in O(g(\delta))$, $p(m+1) \notin O(g(\delta)^{2})$ and $V^{s}(m+1) \notin O(g(\delta)^{2})$, we now consider the case of $p(m)$.

The proof proceeds by writing the Nash bargaining pricing equation, stepping the value functions forward, and rewriting in terms of the next step price $p(m+1)$:

From the proof of \cref{prop:immediate_agreement}, $p(m) = \frac{1}{1 - \delta  \rho(m)} \frac{\kappa(m)}{s(m)}$, so $p(m) > 0$ is equivalent to $\kappa(m) > 0$, so we will prove the latter.
\begin{align*}
\kappa(m) & = w [1-\delta \hat{\rho}(\delta,m) m]v + w \delta V^{s}(m+1) - \delta^2 w \hat{\rho}(\delta,m) m (n -m - 1)  V^{b}(m+1) \\ 
& \qquad -\delta^2 w \hat{\rho}(\delta,m) m V^{s}(m+1) - (1 - w) \delta V^{s}(m+1) + (1 - w) \delta^2 \hat{\rho}(\delta,m) m (n - m) V^{s}(m+1) \\
& =  w [1-\delta \hat{\rho}(\delta,m) m]v + \delta (2w - 1) V^{s}(m+1) \\
& \qquad -\delta^2 \hat{\rho}(\delta,m) m [ w (n - m - 1) V^{b}(m+1) + (w - (1 - w)(n-m)) V^{s}(m+1) ] \\
& =  w [1-\delta \hat{\rho}(\delta,m) m]v + \delta (2w - 1) \hat{\rho}(\delta,m) [ (n-m) p(m+1) + \delta m (n-m) V^{s}(m+2) ] \\
& \qquad -\delta^2 \hat{\rho}(\delta,m) m [ w (n - m - 1) V^{b}(m+1) + (w - (1 - w)(n-m)) V^{s}(m+1) ] \\
& =  w [1-\delta \hat{\rho}(\delta,m) m(n-m)]v + \delta (2w - 1) \hat{\rho}(\delta,m) (n-m) p(m+1) \\
& \quad - \delta^{2} (2w - 1) \hat{\rho}(\delta,m) m [ V^{s}(m+1) - V^{s}(m+2) ] \\
& \quad + \delta \hat{\rho}(\delta,m) m (n - m - 1)  [ w ( v + \delta V^{s}(m+2) - \delta V^{b}(m+1) ) \\
& \quad - (1 - w) (\delta V^{s}(m+2) - \delta V^{s}(m+1) ) ] \\
& =  w [1-\delta \hat{\rho}(\delta,m) m(n-m)]v + \delta (2w - 1) \hat{\rho}(\delta,m) (n-m) p(m+1) \\
& \quad - \delta^{2} (2w - 1) \hat{\rho}(\delta,m) m [ V^{s}(m+1) - V^{s}(m+2) ] \\
& \quad + \delta \hat{\rho}(\delta,m) m (n - m - 1)  p(m+1)
\end{align*}
We can apply the inductive hypothesis to each term, and it follows that $p(m) \in O(g(\delta))$ and $V^{s}(m) \in O(g(\delta))$, whereas $p(m) \notin O(g(\delta)^{2})$ and $V^{s}(m) \notin O(g(\delta)^{2})$ (the $v$ term alone implies these latter two). We can infer the following equivalent condition for $p(m) > 0$:
\begin{align*}
& w [1-\delta \hat{\rho}(\delta,m) m(n-m)]v + \delta (2w - 1) \hat{\rho}(\delta,m) \{ (n-m) p(m+1) \\
& \quad - m \delta [ V^{s}(m+1) - V^{s}(m+2) ] \} + \delta \hat{\rho}(\delta,m) m (n - m - 1)  p(m+1) > 0 \\
\Leftrightarrow & w [1-\delta \hat{\rho}(\delta,m) m(n-m)]v + \delta 2w  \hat{\rho}(\delta,m) \{ (n-m) p(m+1) - m \delta [ V^{s}(m+1) - V^{s}(m+2) ] \} \\
& \quad - \delta  \hat{\rho}(\delta,m) \{ (n-m) p(m+1) - m \delta [ V^{s}(m+1) - V^{s}(m+2) ] \} \\
& \quad + \delta \hat{\rho}(\delta,m) m (n - m - 1)  p(m+1) > 0
\end{align*}
Using the fact that $m \delta [ V^{s}(m+1) - V^{s}(m+2) ] \leq p(m+1)$ from \eqref{Equation-mdelta}, we can drop the second term to get a sufficient condition:
\begin{multline*}
\Leftarrow w [1-\delta \hat{\rho}(\delta,m) m(n-m)]v - \delta  \hat{\rho}(\delta,m) \{ (n-m) p(m+1) - m \delta [ V^{s}(m+1) - V^{s}(m+2) ] \} \\
+ \delta \hat{\rho}(\delta,m) m (n - m - 1)  p(m+1) > 0
\end{multline*}
Moreover, from the same derivations as leading to \eqref{Equation-mdelta}, we can equivalently substitute $m \delta [ V^{s}(m+1) - V^{s}(m+2) ] = \delta \hat{\rho}(\delta,m) m (n-m) p(m+1) - [ 1 - \delta \hat{\rho}(\delta,m) m (n-m) ] m \delta V^{s}(m+2)$:
\begin{align*}
\Leftrightarrow & w [1-\delta \hat{\rho}(\delta,m) m(n-m)]v -\delta \hat{\rho}(\delta,m) [ 1 - \delta \hat{\rho}(\delta,m) m (n-m) ] m \delta V^{s}(m+2) \\
& - \delta  \hat{\rho}(\delta,m) \{ (n-m) p(m+1) - \delta \hat{\rho}(\delta,m) m (n-m) p(m+1) \} \\
& + \delta \hat{\rho}(\delta,m) m (n - m - 1)  p(m+1) > 0 \\
\Leftrightarrow & w [1-\delta \hat{\rho}(\delta,m) m(n-m)]v -\delta \hat{\rho}(\delta,m) [ 1 - \delta \hat{\rho}(\delta,m) m (n-m) ] m \delta V^{s}(m+2) \\
& - \delta  \hat{\rho}(\delta,m) [1 - \delta \hat{\rho}(\delta,m) m (n-m)] p(m+1)  \\
& + \delta \hat{\rho}(\delta,m) (m - 1) (n - m - 1)  p(m+1) > 0
\end{align*}
The first term is $> 0$ and is $O(g(\delta))$ and not $O(g(\delta)^{2})$. The second two terms are $O(g(\delta)^{2})$. The last term is $\geq 0$ for large enough $\delta$ but may equal $0$ (if $m = 1$). Therefore, for large enough $\delta$ (small enough $\Delta$), the positive term(s) dominate and $p(m) > 0$, completing the inductive step. 

\end{proof}
\begin{singlespace}
    \addcontentsline{toc}{section}{References}
    \bibliographystyle{ecta}
    \bibliography{sellinginfo}
\end{singlespace} 


\end{document}